\documentclass[a4paper,onecolumn,superscriptaddress,11pt]{quantumarticle}
\pdfoutput=1
\usepackage[utf8]{inputenc}
\usepackage[english]{babel}
\usepackage{amsmath}
\usepackage{amssymb}
\usepackage{amsthm}
\usepackage[destlabel]{hyperref}
\usepackage{csquotes}
\usepackage[style=alphabetic]{biblatex}
\usepackage{algorithm2e}
\usepackage{caption}
\usepackage{subcaption}
\usepackage{multirow}
\usepackage{nicefrac}
\usepackage{thmtools} 
\usepackage{thm-restate}
\usepackage[nameinlink, noabbrev, capitalise]{cleveref}
\usepackage{tocbasic}
\DeclareTOCStyleEntry[
  beforeskip=.35em plus 1pt,
  pagenumberformat=\textbf
]{tocline}{section}

\usepackage{enumitem}
\setlist{nosep} 

\newtheorem{theorem}{Theorem}
\counterwithin{figure}{section}
\counterwithin{table}{section}
\counterwithin{theorem}{section}

\newtheorem{lemma}[theorem]{Lemma}
\newtheorem{corollary}[theorem]{Corollary}

\theoremstyle{remark}

\newcommand*{\q}{\mathbb{Q}}
\newcommand*{\Oe}{\mathcal{O}_E}
\newcommand*{\zm}{\zeta_{2^m}}
\newcommand*{\qzm}{\q(\zeta_{2^m})}
\newcommand*{\degK}{\nicefrac{1}{2}\,\mathrm{degree}(E)}
\newcommand*{\tr}{\mathrm{Tr}}
\newcommand*{\trq}{\mathrm{Tr}_{K/\mathbb{Q}}\,}
\newcommand*{\trqe}{\mathrm{Tr}_{E/\mathbb{Q}}\,}
\newcommand{\ket}[1]{|#1\rangle}
\newcommand{\bra}[1]{\langle #1|}
\newcommand*{\ip}[1]{ \langle #1 \rangle } 

\newcommand{\vk}[1]{#1}

\title{Stabilizer operators and Barnes-Wall lattices}
\author{Vadym Kliuchnikov, Sebastian Sch{\"o}nnenbeck }
\date{June 2024}

\def\arraystretch{1.1}

\begin{document}

\maketitle

\begin{abstract}
We give a simple description of rectangular matrices that can be implemented by a post-selected stabilizer circuit.
Given a matrix with entries in dyadic cyclotomic number fields $\mathbb{Q}(\exp(i\frac{2\pi}{2^m}))$, 
we show that it can be implemented by a post-selected stabilizer circuit if it has entries 
in $\mathbb{Z}[\exp(i\frac{2\pi}{2^m})]$ when expressed in a certain non-orthogonal basis.
This basis is related to Barnes-Wall lattices.
Our result is a generalization to a well-known connection between Clifford groups and Barnes-Wall lattices.
We also show that minimal vectors of Barnes-Wall lattices are stabilizer states, which may be of independent interest.
Finally, we provide a few examples of generalizations beyond standard Clifford groups.
\end{abstract}

\tableofcontents

\section{Introduction}

Our goal is to give a simple description for rectangular matrices that can be implemented by a post-selected stabilizer circuit.
\vk{Briefly, a post-selected stabilizer circuit is a quantum circuit that consists of a sequence of allocations of ancillary qubits in a zero state, Clifford unitary gates
and destructive post-selected computational-basis measurements (see \cref{sec:stabilizer-operators} for the additional background).
These circuits describe an effect of a stabilizer circuit on quantum state, when certain quantum measurement outcomes have been observed.
Stabilizer circuits correspond to operations that can be implemented on fault-tolerant quantum computers relatively cheaply.}
It is known that any $n$-qubit unitary with entries in $\mathbb{Z}[i,1/\sqrt{2}]$ can be implemented by 
a Clifford+T circuit \cite{Giles2013}. We give a similar description for rectangular matrices implemented by a post-selected stabilizer circuit,
except we require to re-express the matrix in a certain basis.
For one qubit, this basis is given by two by two matrix $B$, that we call Barnes-Wall basis matrix:
\begin{equation} \label{eq:b-matrix}
B = 
\left(\begin{array}{cc}
1+i & 1 \\
0   & 1
\end{array}\right)
\end{equation}
In follow-up work \cite{k2024}, we show that our result has numerous application to the synthesis of quantum circuits over Clifford+T and other related gate sets.   
The following theorem is the formal version of the main result of this paper.

\begin{figure}[htp]
    \centering





\includegraphics[]{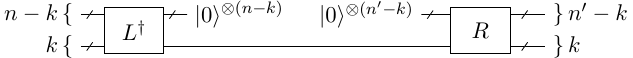}
    \caption{Post-selected general form circuit with $n$ input qubits and  $n'$ output qubits.
    Unitaries $L$ and $R$ are Clifford unitaries, $k$ is the number of inner qubits.
    \vk{The circuit consist of four steps: applying the Clifford unitary $L^\dagger$ to all qubits, applying a post-selected destructive measurement to first $n-k$ qubits, allocation of $n-k'$ qubits in a $\ket{0}$ state, 
    applying the Clifford unitary $R$ to all qubits. See \cref{sec:stabilizer-operators} for additional background.}
    Any post-selected stabilizer circuit implements a linear map proportional to a linear map implemented by some post-selected general form circuit \cite{kliuchnikov2023stabilizer}.}
    \label{fig:stab-circuit}
\end{figure}

\begin{restatable*}[Basis and trace condition]{theorem}{basiscondition}
\label{thm:basis-condition}
Consider a $2^{n'}$ by $2^{n}$ matrix $A$ with entries in $E = \qzm$~($m \ge 2$) such that $\tr(A^\dagger A) = 2^{n'}$ and such that the
matrix $((B^{-1})^{\otimes n'}) A (B^{\otimes n})$ has entries in $\Oe$, the ring of integers of $E$.
There exist unitaries $L,R$ from the $n$-qubit, $n'$-qubit Clifford group, an integer $k \le \min(n,n')$ 
and an integer $j$ such that 
$$
 A = \zm^j  (1+i)^{n'-k} R \cdot \left(\ket{0}^{\otimes (n'-k)}\otimes I_{2^k}\right) \cdot \left(\bra{0}^{\otimes (n-k)}\otimes I_{2^k} \right)\cdot L^\dagger,
$$
that is $A$ is a linear operator that can be implemented by a post-selected stabilizer circuit in \cref{fig:stab-circuit}.
When $A$ is an isometry, $k=n$ and no post-selection is necessary.
\end{restatable*}

An important special case of the above result is a simple description 
of the stabilizer states: 
\begin{corollary}
Consider an $n'$-qubit quantum state $\ket{\psi}$ with entries in $E = \qzm$~($m \ge 2$)
such that $((1+i)B^{-1})^{\otimes n'} \ket{\psi}$ has entries in $\Oe$, the ring of integers of $E$.
There exist a unitary $R$ from the $n'$-qubit Clifford group and an integer $j$ such that 
$ \ket{\psi} = \zm^j  R \ket{0}^{\otimes n'} $, that is, $\ket{\psi}$ is a stabilizer state up to a global phase.
\end{corollary}

We recall the definitions of cyclotomic fields $\qzm$ and their rings of integers in \cref{sec:cyclotomic-fields}. 
The simplest example of a field $E$ and its ring of integers $\Oe$ is the case $m=2$:
$$
E = \q(i) = \{ \alpha + i \beta : \alpha,\beta \text{ are rationals } \} \quad \Oe = \mathbb{Z}[i] = \{ a + i b : a,b \text{ are integers } \}
$$
We carefully choose our definition of the $n$-qubit Clifford group~(\cref{sec:clifford-groups}) so that it is a finite group 
that consists of matrices with entries in $\q(i)$.

Our result can be viewed as a generalization of the fact that Clifford groups are automorphism groups of Barnes-Wall lattices~\cite{Nebe2001}.
More precisely, \cref{thm:basis-condition} for the case when $A$ is a unitary and $m=3$ is equivalent to
the statement of Proposition~6.4~\cite{Nebe2001} given without a proof.
The proof is omitted because the result of Proposition~6.4~\cite{Nebe2001} is similar to the result for the real Clifford groups in Proposition~5.3~\cite{Nebe2001}.
Our proof of \cref{thm:basis-condition} uses techniques familiar to a quantum computing audience and is based on the connection between minimal vectors 
of Barnes-Wall lattices and stabilizer states, that may be of an independent interest (see \cref{sec:preliminaries} for the relevant definitions):
\begin{restatable*}[Minimal vectors]{theorem}{minimalvectors}
\label{thm:minimal-vectors}
The minimal vectors of the $n$-qubit Barnes-Wall lattice over the cyclotomic field $E = \qzm$~($m \ge 2$) are stabilizer states. 
Furthermore, these minimal vectors can always be represented as $\zm^k \cdot (1+i)^n \cdot C\ket{0}^{\otimes n}$ 
for some unitary $C$ from the Clifford group and some integer $k \in [2^m]$.
\end{restatable*}

\vk{
The above result is closely related to several well-known results.
First, it is well-known that the automorphism group of a Barnes-Wall lattice over $\q$ acts transitively on the minimal vectors of the Barnes-Wall lattice~\cite{Broue1973}.
We are not aware of a reference with an explicit proof for Barnes-Wall lattices over $\q(i)$ and complex Clifford groups.
Second, it is known that any stabilizer state can be prepared using a Clifford unitary starting from standard basis states~\cite{Calderbank1998}.
Third, it might suffice to show the result for $E = \q(i)$ and then extend it to the other cyclotomic fields by 
using the techniques related to the tensor products of Hermitian lattices~\cite{coulangeon2012unreasonable,Coulangeon2000}.
Moreover, we expect that the result can be extended to other number fields that contain $\q(i)$ as a subfield using the results on the tensor products of Hermitian lattices.
In this work we provide self-contained proofs using elementary techniques familiar to a quantum computing audience.}

In the appendix we discuss computational methods that allow to 
obtain similar results beyond Clifford groups.
We discuss a heuristic approach that helps finds a matrix $\tilde B$ similar to $B$
for any finite group of unitaries matrices with entries in some number field $\tilde E$ 
so that matrices from the finite group can be characterized as matrices $U$ such that $\tilde B^{-1} U B$ has entries in the ring of integers $O_{\tilde E}$.
For this we rely on the standard algorithm for computing automorphism groups of Hermitian lattices.
We provide a Jupyter notebook using the Julia programming language that implements our heuristic approach.
We provide examples where our heuristic approach succeeds as \cref{thm:state-basis-beyond-clifford,thm:unitary-basis-beyond-clifford} and also comment when it fails.

\section{Preliminaries}
\label{sec:preliminaries}

\subsection{Cyclotomic number fields}
\label{sec:cyclotomic-fields}

Let us recall some basic definitions and facts about cyclotomic number fields.
We denote the $k$-th \emph{root of unity} as $\zeta_k = e^{i \frac{2\pi}{k}}$ and the corresponding \emph{cyclotomic field} as\footnote{For the purposes of this paper we identify cyclotomic fields with a subset of the field of complex numbers.}: 
$$
\q(\zeta_k) = \left\{ \sum_{j \in [d_k]} a_{j} \zeta_k^j \text{ where } a_j \in \q \text{ for } j \in [d_k] \right\},
$$
where $d_k$ is Euler's totient function that counts the number of positive integers less than $k$ that are co-prime to $k$.
The \emph{ring of integers} of a cyclotomic field $\q(\zeta_k)$ is 
$$
\mathbb{Z}[\zeta_k] = \left\{ \sum_{j \in [d_k]} z_{j} \zeta_k^j \text{ where } z_j \in \mathbb{Z} \text{ for } j \in [d_k] \right\}.
$$
An \emph{embedding} is a map from the field $\q(\zeta_k)$ to $\mathbb{C}$ that respects all the field operations, such as addition, multiplication, negation and inverse.
There are $d_k$ distinct embeddings of $\q(\zeta_k)$ into the complex numbers $\mathbb{C}$ that each correspond to 
a positive integer $j$ co-prime to $d_k$. 
For each such $j$, the embedding $\sigma$ is defined by equation $\sigma(\zeta_k) = \zeta_k^j$.
Each such embedding $\sigma$ also corresponds to an \emph{automorphism} of $\q(\zeta_k)$,
that is a map from $\q(\zeta_k)$ to $\q(\zeta_k)$ that respects all the field operations.

For the main text of the paper, we focus on the cyclotomic fields $\qzm$ for integer $m \ge 2$ that include the fourth cyclotomic field $\q(\zeta_4)=\q(i)$ as a subfield. 
We use $E$ to denote a cyclotomic field $\qzm$ and $\Oe$ for its ring of integers.
The real subfield of $\q(\zeta_k)$, also known as a \emph{real cyclotomic field}, is $\q(\zeta_k) \cap \mathbb{R}$.
We use $K$ to denote the real subfield of $E$.
For example, the real subfield of $\q(i)$ is the field of rational numbers $\q$.
It is always the case that $\mathrm{degree}(K) = \degK$.
We use $d = 2^{m-2}$ to denote the degree of $K$.
Note that we can always order the $2d$ embeddings $\sigma_1,\ldots,\sigma_{2d}$ of $E$ so that $\sigma_{l}(\alpha) = \sigma_{d+l}(\alpha)^\ast$ for $l \in [d]$,
where $\,^\ast$ denotes complex conjugation.
We will use $\sigma_1,\ldots,\sigma_d$ for the distinct embeddings of $K$ into $\mathbb{R}$ and $E$ into  $\mathbb{C}$.

For elements of $K$ and $E$ we define the \emph{trace} functions from $K,E$ to $\q$ as:
$$
\trq(\alpha) = \sum_{j \in [d]} \sigma_j(\alpha), \quad \trqe(\beta) = \sum_{j \in [2d]} \sigma_j(\beta),
$$
For elements $\alpha$ from $K \cap \q$ and $\beta$ from $E  \cap \q$, the values of the trace are 
$\trq(\alpha) = d \cdot \alpha$, $\trqe(\beta) = 2d \cdot  \beta$. 

It is convenient to think of $E$ as a vector space over $\q$ with inner product $\ip{\alpha, \beta} = \frac{1}{2}\trqe(\alpha \beta^\ast)$
defined via the trace function.
The Euclidean norm squared of elements of $E$ can be expressed using $\trq$ as
$$
\ip{\alpha, \alpha} = \nicefrac{1}{2}\cdot \trqe(\alpha \alpha^\ast) = \trq(\alpha \alpha^\ast)
$$
The elements $\zm^j$ for $j \in [2d]$ form an orthogonal basis of $E$.
The Euclidean norm squared of each of the basis elements is $d$.
The Euclidean norm squared of elements of $\Oe$ is at least $d$, with the value $d$ achieved only on the roots of unity $\zm^j$.

We will use the fact that $1+i$ divides $1 + \zm^j$ if and only if $\zm^j$ is a power of $i$.
This follows from two facts. The first fact is that $1+i$ can be written as a product of $1+ \zm^j$ for
for odd powers of $j$. This is achieved by repeatedly using $(1-(-u)) = (1 + \sqrt{-u})(1 - \sqrt{-u})$ 
for $u = i, \sqrt{-i}$ and further roots of $i$ till the decomposition is reached.
The second fact is that $(1+\zm^j) / (1-\zm^{j'})$ is a unit in $\Oe$ when $j,j'$ are odd \cite{Washington1997}.

\subsection{Hermitian spaces and Hermitian lattices}

For the purposes of this work we define Hermitian spaces and Hermitian lattices over the fields $E = \qzm$
and use the standard inner-product.
For much more general definitions we refer the reader to \cite{Kirschmer2016}.
Let $E^N$ be a vector space of $N$-dimensional vectors with entries in $E$.
For two elements $\ket{\phi}, \ket{\psi} $ of $E^N$ the inner product is given 
by $ \ip{\phi| \psi} = \sum_{j \in N} \ket{\phi}_j \ket{\psi}_j^\ast  $, where $\ket{\phi}_j$ is the $j$-th 
coordinate of $\ket{\phi}$.
The pair $(E^N, \ip{\cdot|\cdot})$ is an example of a \emph{Hermitian space}.

Given an invertible $N \times N$ matrix $M$ \vk{with entries in $E$}, the set $L = M \Oe^N$ is an example of a \emph{totally-definite Hermitian lattice}, 
$M$ is a \emph{basis matrix} of $L$. We will say that $L$ is a lattice over the field $E$.
Given a lattice $L$, its \emph{dual} is
$$
L^\ast = \left\{ \ket{\psi} \in E^N : \ip{ \psi | \phi} \in \Oe \text{ for all } \ket{\phi} \in L \right\}.
$$
Notably, the dual of the dual of the lattice in the lattice itself, $L = (L^\ast)^\ast$.
Given a basis matrix $M$ of $L$, one basis matrix of the dual lattice $L^\ast$ is $(M^{-1})^\dagger$.
The expression for the basis of dual lattice is a generalisation of a similar result for integer lattices that can be found in Ref.~\cite{Regev2019}.
\vk{We have said that $L = M \Oe^N$ is an example of a Hermitian lattice because not every Hermitian lattice has a basis. For the general definition of Hermitian lattices see~\cite{Kirschmer2016}. }

The \emph{minimum} of the lattice $L$ over  the field $E$ with real subfield $K$ is
$$
\min_{\ket{\psi} \in L , \ket{\psi} \ne 0} \trq \ip{\psi | \psi},
$$
and the lattice vectors on which the minimum is achieved are the \emph{minimal vectors}.
Note that a minimal vector multiplied by $\zm^j$ is again a minimal vector. 

\subsection{Pauli matrices, stabilizer states and Clifford groups}
\label{sec:stabilizer-states}

\subsubsection{Pauli matrices}

We use $I,X,Y,Z$\footnote{Notation $\sigma_x,\sigma_y,\sigma_z$ is sometimes used for $X,Y,Z$} for the two by two \emph{Pauli matrices} with entries in $\q(i)$.
A \emph{Pauli operator} $P$ is an element of the set $\pm \{I,X,Y,Z\}^{\otimes n}$.
A set of Pauli operators is \emph{independent} if none of them can be expressed as a product of the other operators from the set.
Two Pauli operators $P,Q$ \emph{commute} when $PQ = QP$.

\subsubsection{Stabilizer states}
A vector $\ket{\psi}$ with entries in $E = \qzm$ is \emph{stabilized} by a Pauli operator $P$ if $P\ket{\psi} = \ket{\psi}$.
We call $P$ a \emph{stabilizer} of $\ket{\psi}$.
If two Pauli operators $P$,$Q$ stabilize a \vk{non-zero} state $\ket{\psi}$, then $P$ and $Q$ commute.
\vk{Indeed, if $P$, $Q$ do not commute then $PQP^\dagger Q^\dagger = -I$ and $\ket{\psi}$ must be a $+1$-eigenstate of $-I$ which is only possible when $\ket{\psi}$ is zero. }
The stabilizer group of a vector $\ket{\psi}$ is 
$$
\mathrm{Stab}(\ket{\psi}) = \left\{ P \in \pm\{I,X,Y,Z\}^{\otimes n} : P\ket{\psi} = \ket{\psi} \right\}
$$
Any $2^n$ dimensional vector that is stabilized by a set of $n$ independent commuting Pauli operators $P$ is a \emph{stabilizer state}.
Note that we do not require that the norm squared $\ip{\psi|\psi}$ is one.
We use the notation $\ket{0} = \left( \begin{array}{c} 
     1\\
     0 
\end{array} \right)$ and $\ket{1}= \left( \begin{array}{c}
     0\\
     1 
\end{array} \right)$.
Both  $\ket{0}$ and  $\ket{1}$ are stabilizer states with $\mathrm{Stab}(\ket{0})= \{I,Z\}$
and $\mathrm{Stab}(\ket{0})= \{I,-Z\}$. The states $\ket{0} \pm \ket{1}$ are stabilizer states with 
stabilizer groups $\{I,\pm X\}$; the states $\ket{0} \pm i\ket{1}$ are stabilizer states with 
stabilizer groups $\{I,\pm Y\}$.
\vk{For additional background on the stabilizer formalism see Section~10.5.1 in~\cite{Nielsen2012}.}

Stabilizer states are exactly the vectors $\ket{\psi}$ with Pauli expectation $\bra{\psi}P\ket{\psi}$ equal to either $0$ 
or $\pm \ip{\psi|\psi}$ for all Pauli operators $P$. 
The next lemma shows that the condition for the Pauli expectation
$\ip{\psi | P | \psi} = \ip{ \psi | \psi }$ can be replaced with $\trq \ip{\psi | P | \psi} = \trq \ip{ \psi | \psi }$
where $K$ is the real subfield of $E$.

\begin{lemma}
\label{lem:trace-stabilizer-condition}
Let $E = \qzm$ be a cyclotomic field that contains $\q(i)$ and let $K$ be the real subfield of $E$.
Consider an $n$-qubit state $\ket{\psi}$ with entries in $E$ such that, for some Pauli operator $P$, 
we have $\trq \ip{\psi | P | \psi} = \trq \ip{ \psi | \psi }$. Then $\ket \psi$ is stabilized by $P$.
\end{lemma}
\begin{proof}
Let $\sigma_1, \ldots, \sigma_d$ be the embeddings of $E$ into $\mathbb{C}$ that are all distinct, 
even when restricted to elements of $K$. 
These embeddings also map elements of $K$ into $\mathbb{R}$ and 
yield exactly the distinct embeddings of $K$ into $\mathbb{R}$.

Let us use the notation $|\psi'\rangle = \sigma_1(|\psi\rangle) \oplus \ldots \oplus \sigma_d(|\psi\rangle)$, 
where $\oplus$ denotes the direct sum and $\sigma_l(|\psi\rangle)$ is the result of element-wise application
of $\sigma_l$ to $|\psi\rangle$
in the computational 
basis\footnote{For example, $\sigma_l(\alpha_0 \ket{0} + \alpha_1 \ket{1}) = \sigma_l(\alpha_0)\ket{0} + \sigma_l(\alpha_1)\ket{1}$}. 
We also denote $P' = \sigma_1(P) \oplus \ldots \oplus \sigma_d(P)$, 
where $\sigma_l$ is applied element-wise to $P$ in the computational basis.
\vk{The matrix $P'$ is a unitary matrix and is diagonal in a basis of eigenvectors according to the spectral theorem.}

Using this notation, we have 
$\trq \ip{ \psi | \psi } = \ip{ \psi' | \psi' } = \trq\ip{ \psi |P | \psi} = \ip{ \psi' | P' | \psi' }$. 
Note that $P'$ has eigenvalues $\pm 1$, the same as the eigenvalues of $P$. 
The vector $\ket{\psi'}$ must be the $+1$-eigenvector of $P'$ so that the equality $\ip{ \psi' | \psi' } = \ip{ \psi' | P' | \psi' }$ is possible. The equality $P' \ket{\psi'} = \ket{\psi'}$ implies $P\ket{\psi} = \ket{\psi}$, which completes the proof.
\end{proof}

The stabilizer group uniquely defines the stabilizer state~\cite{Nielsen2012} up to a scalar.
That is, given two stabilizer states $\ket{\phi}$, $\ket{\psi}$ equality $\mathrm{Stab}\ket{\phi} = \mathrm{Stab}\ket{\psi}$
implies that there exist $\alpha \in E$ such that $\ket{\phi} = \alpha \ket{\psi}$.

\subsubsection{Clifford groups}
\label{sec:clifford-groups}

Given a $2^n \times 2^n$ matrix $M$ with entries in the field $E = \qzm$, we use the notation $M^\dagger$ for the conjugate-transpose of $M$.
We say that $M$ is a unitary when $MM^\dagger$ is the identity matrix.
A Clifford unitary is a unitary with entries in $E$ such that $C P C^\dagger \in \pm \{I,X,Y,Z\}^{\otimes n}$ for all Pauli operators $P$.
Trivially, for any $u \in E$ with $|u|^2 = 1$, the global phase unitary $u I^{\otimes n}$ is an $n$-qubit Clifford unitary.
Any $n$-qubit Clifford unitary can be written as a product of a global phase unitary
and transvections $\frac{1+i}{2} I^{\otimes n} + \frac{1-i}{2} P$ for $P \in \pm\{I,X,Y,Z\}^{\otimes n}$.
This follows from the connection between Clifford unitaries and binary symplectic groups.
The $n$-qubit \emph{Clifford group} is the group generated by  $\frac{1+i}{2} I^{\otimes n} + \frac{1-i}{2} P$ for $P \in \pm\{I,X,Y,Z\}^{\otimes n}$.
Another useful set of generators is $\mathrm{CNOT}$, $S = diag\{1,i\}$ and the
global-phase-adjusted Hadamard gate $\tilde H = \frac{1}{1+i}\left( \begin{array}{cc}
    1 & 1 \\
    1 & -1
\end{array} \right)$~\cite{Dehaene2003,Aaronson2004}.

The Clifford group acts transitively on stabilizer states \cite{Dehaene2003,Aaronson2004}, that is for any stabilizer state $\ket{\psi}$
there exist a Clifford unitary $C$ and a scalar $\alpha \in E$ such that $\ket{\psi} = \alpha C\ket{0}^{\otimes n}$.
The transitive action of the Clifford group is useful for establishing the following:

\begin{lemma}
\label{lem:up-to-sign-equality}
Let $\ket{\phi}$, $\ket{\psi}$ be two stabilizer states over $E = \qzm$ such that for every Pauli operator $P$ that stabilizes 
$\ket{\phi}$ the Pauli operator $\pm P$ stabilizes $\ket{\psi}$.
Then there exist a scalar $\alpha \in E$ and a Pauli operator $Q$ such that $\ket{\phi} = \alpha Q \ket{\psi}$.
\end{lemma}
\begin{proof}
First, reduce the general case to $\ket{\phi} = \ket{0}^{\otimes n}$ using that the Clifford group acts transitively on stabilizer states.
In this special case $\ket{\psi}$ must be proportional to $\ket{a_1} \otimes \ldots \otimes \ket{a_n}$ for $a_k \in \{0,1\}$.
For this reason $\ket{\phi} = \alpha Q \ket{\psi}$ for some $Q \in \{I,X\}^{\otimes n}$.
\end{proof}


\subsection{Stabilizer operators and post-selected stabilizer circuits}
\label{sec:stabilizer-operators}

\vk{A post-selected stabilizer circuit consist of three kinds of operations: allocation of ancillary qubits in a $\ket{0}$ state,
applying Clifford unitaries to a subset of qubits and post-selected computational basis measurements.
We say that an operations applies a linear operator $A$, when it transform an input state $\ket{\psi}$ to the output state $A \ket{\psi}$.
For example, allocating a qubit in a zero state in addition to existing $n$ qubits corresponds to applying the operator $\ket{0} \otimes I_{2^n}$,
where we treat $\ket{0}$ as a $2 \times 1$ matrix.
Similarly, applying a Clifford unitary $U$ to the first two out of $n$ qubits corresponds to applying the operator $ U \otimes I_{2^{n-2}}$.
A post-selected measurement of $\ket{0}$ state on the first out of $n$ qubits corresponds to applying operator $\bra{0} \otimes I_{2^{n-1}}$,
where we treat $\bra{0}$ as a $1 \times 2$ matrix.
For more background on quantum circuits see \cite{Nielsen2012}.}
 
A \emph{stabilizer operator} with $n$ input qubits and $n'$ output qubit and $k \le \min(n,n')$ inner qubits is 
a following $2^{n'} \times 2^n$ matrix with entries in $\q(i)$ 
\begin{equation}
\label{eq:stab-operator}
      A = (1+i)^{n'-k} R \cdot \left(\ket{0}^{\otimes (n'-k)}\otimes I_{2^k}\right) \cdot \left(\bra{0}^{\otimes (n-k)}\otimes I_{2^k} \right)\cdot L^\dagger,
\end{equation}
where $R$ is an element of the $n'$-qubit Clifford group, $L$ is an element of the $n$-qubit Clifford group, and $I_{2^k}$ is the $2^k$
dimensional identity matrix.
\vk{\cref{fig:stab-circuit} visualizes \cref{eq:stab-operator} as a post-selected quantum circuit. }  
Any post-selected stabilizer circuit with $n$ input qubits and $n'$ output qubits applies a linear map proportional to a stabilizer operator~\cite{kliuchnikov2023stabilizer}.

Stabilizer operators are closely related to stabilizer states \vk{via Choi states}.
\vk{Given a $2^{n'} \times 2^n$ linear operator $B$ the corresponding \emph{Choi state} is the following $2^{n+n'}$ dimensional vector
\begin{equation}
\label{eq:choi-state-def}
     \ket{B} = (1+i)^n \sum_{ j \in \{0,1\}^n } \ket{j}\otimes B\ket{j}.
\end{equation}
The motivation for the choice of the scalars for the Choi state and stabilizer operators will become clear later in \cref{sec:stab-states-as-minimal-vectors}.
If $\ket{B}$ is a stabilizer state, than $B$ is proportional to a stabilizer operator~\cite{kliuchnikov2023stabilizer,Audenaert2005}.
This is a direct consequence of a bipartite normal form for stabilizer states used to study their entanglement~\cite{Audenaert2005}.}
\vk{For more background on the Choi states, see ``Operator-vector correspondence'' in Section~1.1.2 \cite{Watrous2018}.}

\vk{There is an efficient algorithms to find $R$, $L$ and $k$ in \cref{eq:stab-operator}  given a Choi state $\ket{A}$.
First, there is an efficient algorithm to find a set of independent generators of the stabilizer group of a stabilizer state~\cite{Silva2023}.
Second, there is an efficient algorithm to find $R$, $L$ and $k$ given the generators of the stabilizer group~\cite{Audenaert2005,kliuchnikov2023stabilizer}.}


\subsection{Barnes-Wall lattices}
\label{sec:bw-lattices}
The one qubit Barnes-Wall lattice over a cyclotomic field $E = \qzm$ for $m \ge 2$
is the lattice with basis matrix 
\begin{equation}
\label{eq:basis}
B = 
\left(\begin{array}{cc}
1+i & 1 \\
0   & 1
\end{array}\right).
\end{equation}
By the definition of the lattice with basis matrix $B$, it is the set $B \Oe^2$
and can also be written as 
$$
\left\{ z_1 \cdot (1+i)\ket{0} + z_2 \cdot (\ket{0} + \ket{1}) : z_1,z_2 \in \Oe \right\}.
$$
The vectors $(1+i)\ket{0}$ and $\ket{0} + \ket{1}$ are the basis vectors 
of the one qubit Barnes-Wall lattice.
The one qubit Barnes-Wall lattice also contains the vectors
$$
(1+i)\ket{1} = (1+i)(\ket{0}+\ket{1}) - (1+i)\ket{1} \quad \ket{0}-\ket{1} = \ket{0} + \ket{1} - (1-i)\cdot(1+i)\ket{1}.
$$
The $n$-qubit Barnes-Wall lattice over a cyclotomic field $E = \qzm$ is the lattice 
over $E$ with basis matrix $B^{\otimes n}$. In other words, it is the set 
$B^{\otimes n} \Oe^{2^n}$. The basis vectors of the $n$-qubit Barnes-Wall lattice are
$$
\ket{\phi_1} \otimes \ldots \otimes \ket{\phi_n},\text{ where } \ket{\phi_s} \in \left\{ (1+i)\ket{0}, \ket{0}+\ket{1} \right\}
$$

\subsubsection{Lower-bound on the minimum}
\label{sec:minima-lower-bound}
The minimum of the $n$-qubit Barnes-Wall lattice over $E$ is at most 
$\degK \cdot 2^n$.
This is because the norm squared of all basis vectors (in our choice of basis) of the lattice is $2^n$
and $\trq(2^n) = 2^n \cdot \mathrm{degree}(K) = 2^n \cdot \degK$, where $K$ is the real subfield of $E$.
Later in \cref{sec:stab-states-as-minimal-vectors} we will establish that 
this bound on the minimum is tight and basis vectors are minimal vectors of the lattice.

\subsubsection{Dual lattice}
\label{sec:dual-lattice}
The basis matrix of the dual of the one qubit Barnes-Wall lattice is $(B^{\dagger})^{-1}$. 
However, it is more convenient to choose another basis matrix for the dual  
equal to $\frac{1}{1+i} B$.
One can check that $(B^{\dagger})^{-1}$ and $\frac{1}{1+i} B$ are bases of the same lattice by direct calculation.
The basis matrix of the dual of the $n$-qubit Barnes-Wall lattice is $\frac{1}{(1+i)^n} B^{\otimes n}$.
This shows that Barnes-Wall lattices are modular lattices.

\subsubsection{Relation to Clifford groups}

The Clifford group~(as defined in \cref{sec:clifford-groups}) preserves the Barnes-Wall lattice, i.e. the unitaries from the group map vectors from the Barnes-Wall lattice to other vectors 
from the Barnes-Wall lattices.
Given an $n$-qubit unitary $U$ with entries in $E=\qzm$, one can check that it preserves the Barnes-Wall lattice
by checking that the matrix $(B^{\otimes n})^{-1} U (B^{\otimes n})$\footnote{It is possible to compute this expression faster than by using naive matrix multiplication using an approach similar to the Fast Hadamard Transform. See also \cite{Micciancio2008}.} has entries in $\Oe$.
In particular, we can check that this is the case for $\mathrm{CNOT}$, $S$ and $\tilde H$
for the two-qubit and one-qubit Barnes-Wall lattices over $\q(i)$.
The tensor product structure of the basis matrix of Barnes-Wall lattices and the fact that $\mathrm{CNOT}$, $S$ and $\tilde H$
generate the $n$-qubit Clifford group implies that Clifford unitaries preserve Barnes-Wall lattices.

\section{Minimal vectors of Barnes-Wall lattices and stabilizer states}
\label{sec:stab-states-as-minimal-vectors}

The goal of this section is to prove the following result.

\minimalvectors*

Our proof strategy is as follows.
We first show that minimal vectors of the $n$-qubit Barnes-Wall lattice over the cyclotomic field $E = \qzm$~($m \ge 2$) 
are very similar to minimal vectors of the $n$-qubit Barnes-Wall lattice over $\q(i)$~(\cref{lem:minimal-vectors-phases}).
Next we establish that the upper-bound on the minimum of Barnes-Wall from \cref{sec:minima-lower-bound} is tight and 
that the minimal vectors on the $n$-qubit Barnes-Wall lattices can be expressed in terms of minimal vectors of the $(n-1)$-qubit Barnes-Wall lattice~(\cref{lem:minima}).
Finally, we establish the main result~(\cref{thm:minimal-vectors}) by induction on the number of qubits $n$ and by relying on some common properties of stabilizer states~(\cref{lem:orthogonal-stab-states}).

The following lemma illuminates the crucial interaction between properties of Barnes-Wall lattices
and properties of cyclotomic fields $\qzm$. 

\begin{lemma}[Phases in minimal vectors]
\label{lem:minimal-vectors-phases}
\textbf{If}  a vector $(1+i)^{n-1} \cdot (\alpha \ket{0} + \beta \ket{1}) \otimes \ket{0}^{\otimes (n-1)}$ 
is a minimal vector of the $n$-qubit Barnes-Wall lattice over the field $E = \qzm$~($m\ge2$), and
$\alpha$, $\beta$ are non-zero elements of $E$,
\textbf{then} $\alpha = \zm^k$ and $\beta = i^j \zm^k$ for some integers $j \in [4]$ and $k \in [2^m]$.
\end{lemma}
\begin{proof}
The proof consists of three steps.
In the first step, we show that $\alpha$ and $\beta$ must be integers from $\Oe$.
In the second step, we show that $\alpha$ and $\beta$ must be powers of $\zm$.
In the third step, we show that $\beta = i^j \alpha$ for some integer $j \in [4]$.
We use $K$ to denote the real subfield of $E$ and $\ket{\phi} = (1+i)^{n-1} \cdot (\alpha \ket{0} + \beta \ket{1}) \otimes \ket{0}^{\otimes (n-1)}$ in what follows.

In the first step, we show that $\alpha$ and $\beta$ must be integers from $\Oe$ by computing the inner product with some of the vectors from the dual lattice of the Barnes-Wall lattice.
Recall that the vectors $\ket{\phi_0^\#} = \frac{1}{(1+i)^{n-1}}\ket{0}\otimes(\ket{0}+\ket{1})^{\otimes (n-1)}$ and $\ket{\phi_1^\#} = \frac{1}{(1+i)^{n-1}}\ket{1}\otimes(\ket{0}+\ket{1})^{\otimes (n-1)}$ belong to the dual lattice of the Barnes-Wall lattice (\cref{sec:dual-lattice}). 
For this reason, the inner products between $\ket{\phi}$ and $\ket{\phi_0^\#}$, $\ket{\phi_1^\#}$ must be integers from $\Oe$. These inner products are exactly $\alpha$ and $\beta$, therefore $\alpha$ and $\beta$ are integers from $\Oe$.

In the second step, we show that $\alpha$ and $\beta$ must be roots of unity by using the upper bound on the minimum of the Barnes-Wall lattice. Indeed:
$$
\trq \ip{ \phi | \phi } = 2^{n-1} \cdot \left( \trq \alpha \alpha^\ast + \trq \beta \beta^\ast \right)
$$
As discussed in \cref{sec:cyclotomic-fields}, for non-zero $\alpha$, $\beta$ from $\Oe$ we have
\begin{equation}
\label{eq:trq-oe-bound}
\trq \alpha \alpha^\ast \ge \degK,\quad \trq \beta \beta^\ast \ge \degK.
\end{equation}
Therefore $ \trq \ip{ \phi | \phi } \ge \degK \cdot 2^n$.
Recall that the minimum of the Barnes-Wall lattice is at most $\degK \cdot 2^n$ as discussed in \cref{sec:bw-lattices}, therefore $\alpha$ and $\beta$ must achieve equality in \cref{eq:trq-oe-bound}. This is only possible when $\alpha$, $\beta$ are roots of unity as discussed in \cref{sec:cyclotomic-fields}.

In the third step, we show that $\alpha = i^j \beta$ for some integer $j \in [4]$.
We compute the inner product of $\ket{\phi}$ with another basis vector of the dual of the Barnes-Wall lattice.
Consider the following basis vector of the dual lattice:
$$
\ket{\phi^\#} = \frac{1}{(1-i)^n} (\ket{0} + \ket{1})^{\otimes n}.
$$
Note that the inner product $\ip{\phi^\# | \phi} = (\alpha + \beta)/(1+i)$ must be an element of $\Oe$.
It remains to show that $1+i$ divides $\alpha + \beta$ only when $\beta = i^j \alpha$.
We write $\beta = u \alpha$ for $u$ being some power of $\zeta_{4m}$.
Since $\alpha$ is a unit, $1+i$ must divide $1+u$.
As discussed in \cref{sec:cyclotomic-fields}, this is only possible when $u$ is a power of $i$, which completes the proof.
\end{proof}

The above lemma shows that at least some of the minimal vectors of the Barnes-Wall lattices over $\qzm$
are the same up to multiplication by roots of unity $\zm$. Next we see that this is the case for all minimal 
vectors and prove that the bound on the minimum of Barnes-Wall lattices given in \cref{sec:minima-lower-bound} is tight.

\begin{lemma}[The minimum of Barnes-Wall lattices]
\label{lem:minima}
The minimum of the $n$-qubit Barnes-Wall lattice over the cyclotomic field $E = \qzm$~($m \ge 2$)
is $\degK \cdot 2^n$.
The minimal vectors of the one qubit Barnes-Wall lattice are 
\begin{equation} \label{eq:minimal-vectors-one-qubit}
\zm^k (1+i) \ket{0},\quad
\zm^k  (1+i) \ket{1},\quad
\zm^k (\ket{0} + i^j\ket{1}) \quad
\text{for } j \in [4], k \in [4m],
\end{equation}
the minimal vectors of the $n$-qubit Barnes-Wall lattice can always be written as 
\begin{equation} \label{eq:minimal-vectors-property}
\ket{0} \otimes \ket{v_0} + \ket{1} \otimes \ket{v_1},
~(1+i)\ket{0}\otimes \ket{v_0},
~(1+i)\ket{1}\otimes \ket{v_1},
\end{equation}
where $\ket{v_0}$, $\ket{v_1}$ are minimal vectors of the $(n-1)$-qubit Barnes-Wall lattice.
\end{lemma}
\begin{proof}
First, recall that the minimum of the $n$-qubit Barnes-Wall lattice is at most $\degK \cdot 2^n$, 
as discussed in \cref{sec:minima-lower-bound}. 
Next, we show that the minimum of the $n$-qubit Barnes-Wall lattice is at least $\nicefrac{1}{2}~\mathrm{degree}(E) \cdot 2^n$ by induction on $n$,
and at the same time establish \cref{eq:minimal-vectors-one-qubit,eq:minimal-vectors-property}.
As before we use $K$ for the real subfield of $E$.

Consider the base case of the induction proof, where $n = 1$. 
Any one-qubit Barnes-Wall lattice vector can be written in one of the following three ways:
$$
\ket{u_0} = \alpha \ket{0},\quad
\ket{u_1} = \beta \ket{1},\quad
\ket{u_{01}} = \alpha' \ket{0} + \beta' \ket{1}
$$
for non-zero $\alpha,\beta,\alpha',\beta'$ from $E$.
We establish the lower bound on $\trq$ of the norm squared of the lattice vectors in each of the above cases separately:
First consider a lattice vector $\ket{u_0} = \alpha \ket{0}$ proportional to $\ket{0}$. To show that $\trq\langle u_0 | u_0 \rangle = \trq |\alpha|^2 \ge \degK \cdot 2$, we note that the vector $\frac{1}{1+i}(\ket{0} + \ket{1})$ belongs to the dual lattice of the one-qubit Barnes-Wall lattice. Therefore, its inner product with $\alpha \ket{0}$ must be a cyclotomic integer from $\Oe$. This implies that $\alpha$ must be equal to $(1-i)z$ for some $z$ from $\Oe$, and $|\alpha|^2 = 2|z|^2$. As noted in \cref{sec:cyclotomic-fields}, $\trq(|z|^2) \ge \degK$ with the equality achieved when $z$ is a root of unity.
This shows the lower-bound on $\trq\langle u_0 | u_0 \rangle$.
Similarly, we establish the lower bound for the lattice vectors proportional to $\ket{1}$.
Consider now a lattice vector $\ket{u_{01}} = \alpha'\ket{0} + \beta' \ket{1}$ for non-zero $\alpha',\beta'$. Applying \cref{lem:minimal-vectors-phases} with $n=1$ shows that $\alpha',\beta'$ have the form described in \cref{eq:minimal-vectors-one-qubit}. This also establishes the required lower bound on the lattice minima for $n=1$.

Suppose now that we have established the result for the $(n-1)$-qubit Barnes-Wall lattices.
Any $n$-qubit Barnes-Wall lattice vector can be written as
$$
\ket{u'_0} = (1+i)\ket{0} \otimes \ket{v_0},\quad
\ket{u'_1} = (1+i)\ket{1} \otimes \ket{v_1},\quad
\ket{u'_{01}} = \ket{0} \otimes \ket{v'_0} + \ket{1} \otimes \ket{v'_1}
$$
for non-zero vectors $\ket{v_0},\ket{v_1},\ket{v'_0},\ket{v'_1}$.
We establish the lower bound on $\trq$ of the norm squared of the lattice vectors in each of the above cases separately.
Consider $\ket{u'_0}$ and show that $\ket{v_0}$ is in the $n-1$ qubit Barnes wall lattice. 
Consider the vector $\frac{1}{1+i}(\ket{0}+\ket{1}) \otimes \ket{w^\#}$ from the dual $n$-qubit 
Barnes-Wall lattice, where $\ket{w^\#}$ is any basis vector of the dual of the $(n-1)$-qubit Barnes-Wall lattice. 
The inner product between $(1+i)\ket{0} \otimes \ket{v_0}$ and $\frac{1}{1+i}(\ket{0}+\ket{1}) \otimes \ket{w^\#}$ is in $\Oe$, and therefore the inner product between $\ket{v_0}$ and $\ket{w^\#}$ is in $\Oe$.
This implies that $\ket{v_0}$ is in the dual lattice of the dual of the Barnes-Wall lattice, which is the Barnes-Wall lattice itself. 
This implies that $\trq\langle u'_0 | u'_0 \rangle = 2\trq\langle v_0 | v_0 \rangle \ge \nicefrac{1}{2}~\mathrm{degree}(E) \cdot 2^n$. Similarly consider $\ket{u'_1}$ and show that $\ket{v_1}$ is in the $(n-1)$-qubit Barnes-Wall lattice and 
$\trq\langle u'_1 | u'_1 \rangle \ge \nicefrac{1}{2}~\mathrm{degree}(E) \cdot 2^n$.
Now by considering the inner product of $\ket{u'_{01}}$
with the vectors $\ket{0}\otimes \ket{w^\#}$, $\ket{1}\otimes \ket{w^\#}$ from the dual of the $n$-qubit Barnes-Wall lattice, we see that $\ket{v'_0}$ and $\ket{v'_1}$ must be from the $(n-1)$-qubit Barnes-Wall lattice and 
the lower-bound $\trq\langle u'_{01} | u'_{01} \rangle = \trq\langle v'_{0} | v'_{0} \rangle + \trq\langle v'_{1} | v'_{1} \rangle  \ge \degK \cdot 2^n$ follows.
\end{proof}

An immediate corollary of the above lemma is that the basis we chose for the Barnes-Wall lattice consist of minimal vectors.
Moreover, for any Clifford unitary $C$, the stabilizer state $\zm^j (1+i)^n C\ket{0}^{\otimes n}$ is a minimal vector.
There is a small gap that remains between \cref{eq:minimal-vectors-property}~(which recursively describes the minimal vectors of the Barnes-Wall lattices)
and showing that all the minimal vectors of the Barnes-Wall lattices are stabilizer states.
The following lemma is a key technical argument that closes this gap.

\begin{lemma}[Stabilizer states' superposition]
\label{lem:orthogonal-stab-states}
Let $\ket{v_0}$ and $\ket{v_1}$ be two stabilizer states over $E = \qzm$~($m \ge 2$) with the same norm 
and $\trq \mathrm{Re} \ip{v_0 | v_1}=0$ (where $K$ is the real sub-field of $E$).
If the linear combination $(\ket{v_0} + \ket{v_1})/(1+i)$ is also a stabilizer state, then 
there exists a Hermitian Pauli operator $P$ that does not stabilize $\ket{v_1}$ such that $\ket{v_0}$ is proportional to $P \ket{v_1}$.
\end{lemma}
\begin{proof}
First check that $\trq$ of the norm of $\ket{v_{01}} = (\ket{v_0} + \ket{v_1})/(1+i)$ is the same as
$\trq$ of the norm 
of $\ket{v_0}$ and $\ket{v_1}$ by direct calculation. Indeed:
\begin{equation}
\label{eq:norm}
\begin{array}{rl}
\trq\left( \frac{1}{1-i}(\bra{v_0} + \bra{v_1})\frac{1}{1+i}(\ket{v_0} + \ket{v_1}) \right) 
 & = \\
= \frac{1}{2}(\trq\ip{v_0 | v_0} + \trq\ip{v_1 | v_1}) + \trq \mathrm{Re} \ip{v_0 | v_1} 
 & = \\
= \trq \ip{v_0 | v_0 } = \trq \ip{v_1 | v_1 }
\end{array}
\end{equation}

Next let us show that if the Pauli operator $Q$ stabilizes $\ket{v_0}$ then $\pm Q$ stabilizes $\ket{v_1}$.
This immediately implies that there must exist a Pauli operator $P$ such that 
$\ket{v_0}$ is proportional to $P\ket{v_1}$, as shown in \cref{lem:up-to-sign-equality}.
Suppose the Pauli operator $Q$ stabilizes $\ket{v_0}$ and 
consider the expectation of $Q$ for the state $\ket{v_{01}}$:
\begin{align}
\label{eq:pauli-expectation}
\begin{array}{c}
\frac{1}{1-i}(\bra{v_0} + \bra{v_1}) \frac{1}{1+i} Q(\ket{v_0} + \ket{v_1})
= \\ 
= \frac{1}{2}(\ip{v_0 | v_0} + \ip{v_1 |Q| v_1}) + \mathrm{Re}\ip{v_0 | v_1} \in \{0, \ip{v_1 | v_1} \}
\end{array}
\end{align}
The expectation $\ip{v_{01} |Q| v_{01}}$ must be zero or $\ip{v_{01}|v_{01}}$ because $\ket{v_{01}}$ is a stabilizer state.
Applying $\trq$ to the \cref{eq:pauli-expectation} above, we see that $\trq \ip{v_1 |Q| v_1}$ must be 
$\pm \trq \ip{ v_1 | v_1 }$, which implies 
that $\pm Q\ket{v_1} = \ket{v_1}$ as shown in \cref{lem:trace-stabilizer-condition}.
Note that $P$ cannot stabilize $\ket{v_1}$, because this would contradict \cref{eq:norm}. 
\end{proof}

Now we have all the intermediate results to prove the main result of this section.

\minimalvectors
\begin{proof}
We prove the result by induction on the number of qubits $n$.
The case $n=1$ follows from the description of one qubit minimal vectors in \cref{eq:minimal-vectors-one-qubit} in \cref{lem:minima}
and the examples of stabilizer states in \cref{sec:stabilizer-states}.
Suppose we have established the result for the number of qubits $n$ 
and let us show the result for $n+1$.
\cref{eq:minimal-vectors-property} in \cref{lem:minima} implies,
that any $n+1$ qubit minimal vector can be in one of the three forms.
When the minimal vector is $(1+i)\ket{0}\otimes \ket{v_0}$ or 
$(1+i)\ket{1}\otimes \ket{v_1}$ with $\ket{v_0}$, $\ket{v_1}$
minimal vectors of the $n$-qubit Barnes-Wall lattice, the result immediately follows. 
It remains to show the result for a minimal vector 
$$
\ket{u'_{01}} = \ket{0}\otimes \ket{v_0} + \ket{1}\otimes\ket{v_1}
$$
where $\ket{v_0}$ and $\ket{v_1}$ are minimal vectors of the $n$-qubit Barnes-Wall lattice. 

First show that the vectors $\ket{v^{\pm}_{01}} = (\ket{v_0} \pm \ket{v_1})/(1+i)$ must also be 
minimal vectors of the $n$-qubit Barnes-Wall lattice.
Recall that the vectors $\frac{1}{1+i}(\ket{0}\pm\ket{1}) \otimes \ket{w_0}$
belong to the dual of the $(n+1)$-qubit Barnes-Wall lattice for any $\ket{w_0}$
from the dual of the $n$-qubit Barnes-Wall lattice. 
The inner product between $\ket{u'_{01}}$ and the mentioned family of vectors is
an integer from $\Oe$ and equal to the inner product between $\ket{v^{\pm}_{01}}$ and $\ket{w_0}$. 
The vectors $\ket{v^\pm_{01}}$ belong to the $n$-qubit Barnes wall lattice, 
using the fact that the dual of the dual lattice is the lattice itself.
It remains to establish that the vectors $\ket{v^{\pm}_{01}}$  are indeed minimal.
Consider the norms squared 
$$
\frac{1}{2}\ip{v_0+v_1|v_0+v_1} = \frac{1}{2}(\ip{v_0|v_0} + \ip{v_1|v_1}) + \mathrm{Re}(\ip{v_0|v_1})
$$
and 
$$
\frac{1}{2}\ip{v_0-v_1|v_0-v_1} = \frac{1}{2}(\ip{v_0|v_0} + \ip{v_1|v_1}) - \mathrm{Re}(\ip{v_0|v_1}).
$$
Applying $\trq$ to the equations above and noticing that $\trq$  
of the norm squared of the lattice vectors is lower-bounded by the lattice minima
we see that $\trq \mathrm{Re}(\ip{v_0|v_1})$ must be zero.
Therefore $\ket{v^\pm_{01}}$ are minimal vectors and stabilizer states by the induction hypothesis.
Now using \cref{lem:orthogonal-stab-states} we see that there exists a Pauli operator $P$ such 
that $\ket{v_0}$ is proportional to $P\ket{v_1}$.

So far we have established that there exists $\alpha$ from $E$ such that $\ket{v_1} = \alpha P \ket{v_0}$.
Next we will show that $\alpha$ must be a power of $i$ by relying on \cref{lem:minimal-vectors-phases}.
We do this by showing that $\ket{v^+_{01}}$ is Clifford equivalent to a state discussed in \cref{lem:minimal-vectors-phases}.
Recall that there exists a Clifford $C$ such that 
$\ket{v_0} =  \zeta^{k}_{4m}(1+i)^n C \ket{0}^{\otimes n}$.
We can then write $\ket{v_1} = \alpha \zeta^{k}_{4m}(1+i)^n C \tilde P \ket{0}^{\otimes n} $ for some Pauli operator $\tilde P$. 
Note that $\tilde P$ can always be replaced by $i^{j'} X^{a(1)}\otimes \ldots \otimes X^{a(n)} = i^{j'} X^a$ 
for some bitstring $a$ and integer $j'$.
We see that the stabilizer state $\ket{v^+_{01}}$ is equal to $\zeta^{k}_{4m} (1+i)^{n-1} C (\ket{0}^{\otimes n} + i^{j'} \alpha \ket{a})$.
The state $(1+i)^{n-1} (\ket{0}^{\otimes n} + i^{j'} \alpha \ket{a})$ 
can be transformed into $(1+i)^{n-1} (\ket{0} + i^{j'} \alpha \ket{1}) \otimes \ket{0}^{\otimes (n-1)}$
by a unitary from the Clifford group consisting of a product of SWAP and CNOT gates. 
Applying \cref{lem:minimal-vectors-phases} shows that $(1+i)^{n-1} (\ket{0} + i^{j'} \alpha \ket{1}) \otimes \ket{0}^{\otimes (n-1)}$
is a minimal vector only if $\alpha$ is a power of $i$. Because this state is related to $\ket{v_{01}}$ via 
multiplication by an element of a Clifford group and $\ket{v_{01}}$ is a minimal vector, $\alpha$ must be power of $i$.

We have shown that $\ket{v_0}= i^j P \ket{v_1}$ for some Pauli matrix $P$ and integer $j$
so $\ket{u'_{01}} = \ket{0} \otimes \ket{v_0} + \ket{1} \otimes i^j P \ket{v_0}$. The above state can be obtained 
from $(\ket{0} + \ket{1}) \otimes \ket{v_0}$ by applying the Pauli $I\otimes P$ controlled on the first qubit and $S^j$
to the first qubit.
\end{proof}

\section{Minimal vectors of Barnes-Wall lattices and stabilizer operators}
\label{sec:stabilizer-operators-as-min-vectors}

The goal of this section is to prove the following result 

\basiscondition*

Our proof strategy is a as follows.
First, we show that the Choi states~(\cref{eq:choi-state-def}) of the stabilizer operators correspond to minimal vectors of Barnes-Wall lattices in \cref{lem:stab-op-to-miv-vector}.
Second, we characterize linear operators whose Choi states are minimal vectors of Barnes-Wall lattices via basis transformation and trace conditions in \cref{lem:basis-transfrom-and-trace}.
Finally we combine the two lemmas and \cref{thm:minimal-vectors} to establish the main result of the paper.

\begin{lemma}
\label{lem:stab-op-to-miv-vector}
Let $A$ be a stabilizer operator with $n$ input qubits, $n'$ output qubits,
then its Choi state 
$$
\ket{A} = (1+i)^n \sum_{j \in \{0,1\}^n } \ket{j}\otimes A\ket{j} 
$$
is a minimal vector of the Barnes-Wall lattice over $\q(i)$ on $n+n'$ qubits.
\end{lemma}
\begin{proof}
Using the definition of the stabilizer operator \cref{eq:stab-operator} and equality 
$$
 \sum_{j \in \{0,1\}^n } \ket{j}\otimes L^\dagger \ket{j}  = \sum_{j \in \{0,1\}^n } L^\ast \ket{j}\otimes \ket{j} 
$$
we can rewrite the Choi state as
$$
\ket{A} =  (1+i)^{n+n'-k} (L^\ast \otimes R) \left( \sum_{j \in \{0,1\}^{n} } \ket{j}\otimes  \left(\ket{0}^{\otimes (n'-k)}\otimes I_{2^k}\right) \cdot \left(\bra{0}^{\otimes (n-k)}\otimes I_{2^k} \right) \ket{j} \right).
$$
Using that $\ket{j}\otimes \left(\bra{0}^{\otimes (n-k)}\otimes I_{2^k} \right)\textbf{}\ket{j} = \ket{j}\otimes \left(\otimes I_{2^k} \right)\textbf{}\ket{j}$  is 
non-zero only when the first $n-k$ bits of $j$ are zero, we can further simplify $\ket{A}$ to
$$
(1+i)^{n+n'-k} (L^\ast \otimes R) \left( \ket{0}^{\otimes (n-k)} \otimes \sum_{j' \in \{0,1\}^{k} } \ket{j'}\otimes  \ket{0}^{\otimes (n'-k)} \otimes \ket{j'} \right).
$$
The above vector is a minimal vector of the Barnes-Wall lattice because an element of the $(n+n')$-qubit Clifford group $L^\ast \otimes R$ preserves the lattice
and the vector 
$$
(1+i)^{n+n'-k}\left( \ket{0}^{\otimes (n-k)} \otimes \sum_{j' \in \{0,1\}^{k} } \ket{j'}\otimes  \ket{0}^{\otimes (n'-k)} \otimes \ket{j'} \right).
$$
is a tensor product of $n+n'-2k$ copies of $(1+i)\ket{0}$, and $k'$ copies of $(1+i)(\ket{00} + \ket{11})$.
The vectors $(1+i)\ket{0}$, $(1+i)(\ket{00} + \ket{11})$ are minimal vectors of the one-qubit and two-qubit Barnes-Wall lattices over $\q(i)$.
\end{proof}

\begin{lemma}
\label{lem:basis-transfrom-and-trace}
Given a $2^{n'}\times2^{n}$ matrix $A$ with entries in $E = \qzm$~($m \ge 2$) its Choi state 
$$
\ket{A} = (1+i)^n \sum_{j \in \{0,1\}^n } \ket{j}\otimes A\ket{j} 
$$
is a minimal-vector of a Barnes-Wall lattice over $E$ if and only if the
matrix $((B^{-1})^{\otimes n'}) A (B^{\otimes n})$ has entries in 
$\Oe$ and $\tr(A^\dagger A) = 2^{n'}$.
\end{lemma}
\begin{proof}
The proof proceeds in two steps.
First we show that $\ip{A|A} = 2^{n'+n}$ if and only if $\tr(A^\dagger A) = 2^{n'}$.
Second we show that $\ket{A}$ belongs to the $(n+n')$-qubit Barnes-Wall lattice if and only if the
matrix $((B^{-1})^{\otimes n'}) A (B^{\otimes n})$ has entries in 
$\Oe$. 
When $\ket{A}$ is in the Barnes-Wall lattice, the condition $\ip{A|A} = 2^{n'+n}$ ensures that it is a minimal vector.

For the first step, the relation between the norm squared $\ip{A|A}$ and $\tr(A^\dagger A)$ follows from the following well-known observation:
$$
 \left(\sum_{j' \in \{0,1\}^{\otimes n}} \ket{j'}\otimes A\ket{j'}\right)^\dagger\left(\sum_{j \in \{0,1\}^{\otimes n}} \ket{j}\otimes A\ket{j}\right) = \sum_{j \in \{0,1\}} \ip{j |A^\dagger A | j} = \tr(A^\dagger A).
$$

For the second step, recall that $\ket{A}$ belongs to the Barnes-Wall lattice if and only if $(B^{-1})^{\otimes (n+n')}\ket{A}$ has entries in $\Oe$.
Using
$$
 \sum_{j \in \{0,1\}^n } (B^{-1})^{\otimes (n)}\ket{j}\otimes A\ket{j} =  \sum_{j \in \{0,1\}^n } \ket{j}\otimes A  ((B^{-1})^T)^{\otimes (n)}\ket{j}
$$
we rewrite $(B^{-1})^{\otimes (n+n')}\ket{A}$ as
$$
(1+i)^n \sum_{j \in \{0,1\}^n } \ket{j}\otimes ((B^{-1})^{\otimes n'}) A (((B^{-1})^T)^{\otimes n})\ket{j}.
$$
We have established that $(B^{-1})^{\otimes (n+n')}\ket{A}$ has entries in $\Oe$ 
if and only if $$(1+i)^n ((B^{-1})^{\otimes n'}) A (((B^{-1})^T)^{\otimes n})$$ has entries in $\Oe$.
It remains to observe that $(1+i) (B^{-1})^T = B M$ for some matrix $M$ with entries in $\Oe$ and determinant $1$.
The required result follows from the fact that $ A' = ((B^{-1})^{\otimes n'}) A (B^{\otimes n})$ has entries in 
$\Oe$ if and only if $A' (M^{-1})^{\otimes n}$ has entries in $\Oe$. 
\end{proof}

We have now established all the intermediate result needed to prove the main result of the paper.

\basiscondition
\begin{proof}
According to \cref{lem:basis-transfrom-and-trace}, the Choi state $\ket{A}$ is a minimal vector of the $(n+n')$-qubit Barnes-Wall     
lattice. According to \cref{thm:minimal-vectors} $\ket{A}$ is a stabilizer state.
As discussed in \cref{sec:stabilizer-operators} $\ket{A}$ is proportional to the Choi state of a stabilizer operator 
$$
\ket{A'}  = (1+i)^n \sum_{j \in \{0,1\}^n } \ket{j}\otimes ((B^{-1})^{\otimes n'}) A (((B^{-1})^T)^{\otimes n})\ket{j}.
$$
According to \cref{lem:stab-op-to-miv-vector}, the Choi state $\ket{A'}$
is also a minimal vector of the $(n+n')$-qubit Barnes-Wall lattice.
We have $\ket{A} = \alpha \ket{A'}$ for some $\alpha$ from $E$.
Given the description of minimal vectors in~\cref{thm:minimal-vectors}, the equality  $\ket{A} = \alpha \ket{A'}$ 
is only possible when $\alpha$ is a root of unity, which concludes the proof.
Finally, note that $n=k$ when the kernel of $A$ is trivial.
\end{proof}

\section{Concluding remarks}

Our results and proof techniques can likely be extended to real Clifford groups and real cyclotomic fields.
The analog of Barnes-Wall lattice for qutrits does not have a nice tensor product structure, so generalizing our results to qutrits and beyond will require a different approach.
Our result can likely be extended to more general fields that contain $\q(i)$ as a sub-field.
We provide partial results toward these directions in \cref{app:beyond-clifford}.

Transforming a matrix into Barnes-Wall basis can be viewed as an alternative to channel representation as introduced in~\cite{Gosset2014}.
The denominator exponents of matrices over $\mathbb{Z}[i,1/\sqrt{2}]$ with respect to the Barnes-Wall basis are more closely related to the $T$-count
of a unitary than the denominators exponents in the usual computational basis.
The denominator exponents of quantum states expressed in the Barnes-Wall basis can also be used to define a monotone similar to the dyadic monotone~\cite{Beverland2020} that is more efficient to compute.

We also hope that our results help pave the way for a self-contained and accessible to a broader audience proof of the fact that Clifford groups are maximal finite sub-groups of a unitary group\cite{Nebe2001}.

Approximating a state with a stabilizer state (with distance measured via $l_2$ norm) can be related to the decoding of Barnes-Wall lattices.
The decoding of a Barnes-Wall lattice is a well-studied question outside of the quantum computing community~\cite{Micciancio2008}.
\newpage 

\printbibliography

\newpage

\appendix

\section{Beyond standard Clifford group}
\label{app:beyond-clifford}

Here we describe a heuristic numerical technique that can help find an analog of the matrix $B$ in \cref{thm:basis-condition}
for the case when the matrix $A$ in the theorem is unitary or a state.
Suppose we are given a finite group $G$ of $N$-dimensional unitary matrices with entries in some number field $E$ that is a CM-field.
We can then construct a Hermitian lattice 
\begin{equation}
\label{eq:orbit-lattice}
L =  \sum_{v \in \{g\ket{0}\,:\,g \in G\}} v \Oe 
\end{equation}
The automorphism group of $L$ is the group of unitary matrices with entries in $E$ that preserve $L$.
By construction, $G$ is a sub-group of the automorphism group of $L$.
There are well-known algorithms for computing the automorphism group and minimal vectors of $L$~\cite{Kirschmer2016}.
The implementations of these algorithms are available in a standalone Magma package~\cite{Kirschmer2016code} and in the Hecke~\cite{Nemo.jl-2017} Julia package.
When the automorphism group of the lattice $L$ is the same as $G$ and the lattice $L$ has basis\footnote{When $\Oe$ is not a principal ideal domain, $L$ might not have a basis over $\Oe$} $\tilde B$,
we can describe all elements of $G$ as matrices with entries in $\Oe$ when expressed in the basis $\tilde B$.
We can also compute the minimal vectors and check that they coincide with the group orbit $\{g\ket{0}\,:\,g \in G\}$.
If these two conditions hold, we get an analog to a special case of \cref{thm:basis-condition} when the matrix $A$ is an $N\times N$
unitary or an $N\times 1$ state.
The fact that there are finitely many minimal vectors can be used to check if similar results hold for $N'\times N$ isometries.
Our numerical results are summarised in the following two theorems:

\begin{table}[t]
\centering

\begin{tabular}{|c|c|c|c|c|}
\hline 
\multirow{2}{*}{Group $G$} & Dim. & Basis  & \multirow{2}{*}{Number field $E$} & \multirow{2}{*}{Center $C$}\tabularnewline
 & $N$ & change $\tilde{B}$ &  & \tabularnewline
\hline 
\hline 
Clifford & $2$ & $B_{\mathbb{C},1}$ (\ref{eq:basis-clifford}) & $\mathbb{Q}(\zeta_{4m}),m\in[2,8]$ & $\zeta_{4m}^{j}I,j\in[4m]$\tabularnewline
\cline{2-5} \cline{3-5} \cline{4-5} \cline{5-5} 
group~(\ref{sec:clifford-basis}) & $4$ & $B_{\mathbb{C},2}$ (\ref{eq:basis-clifford}) & $\mathbb{Q}(\zeta_{4m}),m\in[2,8]$ & $\zeta_{4m}^{j}I,j\in[4m]$\tabularnewline
\hline 
\hline 
Real Clifford & $2$ & $B_{\mathbb{R},1}$ (\ref{eq:basis-real-clifford}) & $\mathbb{Q}(\zeta_{8m})\cap\mathbb{R},m\in[1,4]$ & $\pm I$\tabularnewline
\cline{2-5} \cline{3-5} \cline{4-5} \cline{5-5} 
group~(\ref{sec:real-clifford-basis}) & $4$ & $B_{\mathbb{R},2}$ (\ref{eq:basis-real-clifford}) & $\mathbb{Q}(\zeta_{8m})\cap\mathbb{R},m\in[1,4]$ & $\pm I$\tabularnewline
\cline{2-5} \cline{3-5} \cline{4-5} \cline{5-5} 
 & $8$ & $B_{\mathbb{R},3}$ (\ref{eq:basis-real-clifford})& $\mathbb{Q}(\zeta_{8})\cap\mathbb{R}=\mathbb{Q}(\sqrt{2})$ & $\pm I$\tabularnewline
\hline 
\hline 
Rational subgroup & $2$ & $B_{\mathbb{Q},1}$ (\ref{eq:basis-rational-clifford}) & $\mathbb{Q}$ & $\pm I$\tabularnewline
\cline{2-5} \cline{3-5} \cline{4-5} \cline{5-5} 
of Clifford group& $4$ & $B_{\mathbb{Q},2}$ (\ref{eq:basis-rational-clifford}) & $\mathbb{Q}$ & $\pm I$\tabularnewline
\cline{2-5} \cline{3-5} \cline{4-5} \cline{5-5} 
 ~(\ref{sec:rational-clifford-basis})  & $16$ & $B_{\mathbb{Q},4}$ (\ref{eq:basis-rational-clifford-4}) & $\mathbb{Q}$ & $\pm I$\tabularnewline
\hline 
\hline 
Qutrit Clifford & $3$ & $B_{1}^{(3)}$ (\ref{eq:basis-qutrit-clifford-1}) & $\mathbb{Q}(\zeta_{3m}),m\in[1,9]$ & $\pm\zeta_{3m}^{j}I$\tabularnewline
\cline{2-5} \cline{3-5} \cline{4-5} \cline{5-5} 
group~(\ref{sec:qudit-clifford-basis}) & $9$ & $B_{2}^{(3)}$ (\ref{eq:basis-qutrit-clifford-2}) & $\mathbb{Q}(\zeta_{3m}),m\in[1,9]$ & $\pm\zeta_{3m}^{j}I$\tabularnewline
\hline 
\hline 
Qudit Clifford & $5$ & $B_{1}^{(5)}$ (\ref{eq:basis-qupant-clifford})& $\mathbb{Q}(\zeta_{5m}),m\in[1,3]$ & $\pm\zeta_{5m}^{j}I$\tabularnewline
group, $d=5$~(\ref{sec:qudit-clifford-basis}) &  &  &  & \tabularnewline
\hline 
\end{tabular}

    \caption{Basis changes related to various Clifford groups and their subgroups that can be used to identify 
    these group elements and the states that can be prepared by unitaries from these groups. See \cref{thm:unitary-basis-beyond-clifford}, 
    \cref{thm:state-basis-beyond-clifford} for how to use this table.}
    \label{tab:basis-changes}
\end{table}

\begin{theorem}
\label{thm:unitary-basis-beyond-clifford}
Suppose that $U$ is an $N\times N$ unitary matrix with entries in the field $E$ such that $\tilde B^{-1} U B$ has entries in $\Oe$,
where $N,E,\tilde B$ are given in Table~1.
The matrix $U$ is a product of generators of the group $G$ times the scalar matrix from $C$, where $G$, $C$ are given in corresponding rows of \cref{tab:basis-changes}.
That is $U$ is in $G$ up to a global phase.
\end{theorem}

\begin{theorem}
\label{thm:state-basis-beyond-clifford}
Suppose that $\ket{\psi}$ is a state with entries in the field $E$ such that $\tilde B^{-1} \ket{\psi}$ has entries in $\Oe$,
where $N,E,\tilde B$ are given in Table~1.
The matrix $\ket{\psi}$ is a product of generators of group $G$ times the scalar matrix from $C$ times $\ket{0}$, where $G$, $C$ are given in corresponding rows of \cref{tab:basis-changes}.
That is $\ket{\psi}$ can be prepared by elements of $G$ from $\ket{0}$ up to a global phase.
\end{theorem}

At first glance it might be surprising that changing the base ring of lattices $\Oe$ does not change the minimal vectors that much.
However, changing the base ring can be viewed as a tensor product operation of two lattices. 
We refer the interested reader to~\cite{coulangeon2012unreasonable} for an investigation of the minimal vectors of tensor products of lattices.

The rational sub-group of the Clifford group on three qubits is an example where our theorem does not hold.
The automorphism group of the lattice derived from the rational subgroup has size $348\,364\,800$ which is larger than the size of three-qubit Clifford group $92\,897\,280$~(\cref{sec:clifford-group-sizes}).
In \cref{sec:rational-clifford-basis}, we show that the rational sub-group of the Clifford group on three qubits is the intersections of automorphism groups of two lattices.

\subsection{Common groups, their generators and basis change matrices}
\subsubsection{Clifford group}
\label{sec:clifford-basis}
Clifford group on $n$ qubits is the group generated by two-qubit gates $CX,CZ$ 
and one-qubit gates $X,Z,\tilde H,S$ and has center $i^j I$.
The corresponding basis change matrices are:
\begin{equation}
\label{eq:basis-clifford}
B_{\mathbb{C},1} = 
\left[
\begin{array}{cc}
\frac{1}{2} \cdot i + \frac{1}{2} & 0 \\
\frac{1}{2} \cdot i + \frac{1}{2} & 1 \\
\end{array}
\right]
\quad 
B_{\mathbb{C},n} = B_{\mathbb{C},1}^{\otimes n}
\end{equation}
Note that we have re-scaled the basis change for Clifford group in the appendix to make the theorem statements in the appendix more uniform.
\subsubsection{Real Clifford group}
\label{sec:real-clifford-basis}
The real Clifford group on $n$ qubits is the group generated by the two-qubit gates $CX,CZ$ 
and the one-qubit gates $X,Z,H$ and has center $\pm I$.
The corresponding basis change matrices are:
\begin{equation}
\label{eq:basis-real-clifford}
B_{\mathbb{R},1} = \left[
\begin{array}{cc}
\frac{1}{2} \cdot \sqrt{2} & 0 \\
\frac{1}{2} \cdot \sqrt{2} & 1 \\
\end{array}
\right]
\quad 
B_{\mathbb{R},n} = B_{\mathbb{R},1}^{\otimes n}
\end{equation}
\subsubsection{Rational subgroup of the Clifford group}
\label{sec:rational-clifford-basis}
The rational subgroup of the Clifford group on $n$ qubits is the group generated by the two-qubit gates $CX,CZ, H \otimes H$ 
and the one-qubit gates $X,Z$ and has center $\pm I$.

The basis changes for one-, two- and three-qubit rational subgroups of the Clifford group are 

\begin{equation}
\label{eq:basis-rational-clifford}
B_{\q,1}=
\left[
\begin{array}{cc}
1 & 0 \\
0 & 1 \\
\end{array}
\right]
\quad 
B_{\q,2}= 
\left[
\begin{array}{cccc}
\frac{1}{2} & 0 & 0 & 0 \\
\frac{1}{2} & 1 & 0 & 0 \\
\frac{1}{2} & 0 & 1 & 0 \\
\frac{1}{2} & 0 & 0 & 1 \\
\end{array}
\right]
\quad
B_{\q,3}=
\left[
\begin{array}{cccccccc}
\frac{1}{2} & 0 & 0 & 0 & 0 & 0 & 0 & 0 \\
0 & \frac{1}{2} & 0 & 0 & 0 & 0 & 0 & 0 \\
0 & 0 & \frac{1}{2} & 0 & 0 & 0 & 0 & 0 \\
\frac{1}{2} & \frac{1}{2} & \frac{1}{2} & 1 & 0 & 0 & 0 & 0 \\
0 & 0 & 0 & 0 & \frac{1}{2} & 0 & 0 & 0 \\
\frac{1}{2} & \frac{1}{2} & 0 & 0 & \frac{1}{2} & 1 & 0 & 0 \\
\frac{1}{2} & 0 & \frac{1}{2} & 0 & \frac{1}{2} & 0 & 1 & 0 \\
0 & \frac{1}{2} & \frac{1}{2} & 0 & \frac{1}{2} & 0 & 0 & 1 \\
\end{array}
\right]
\end{equation}

\begin{equation}
\label{eq:basis-rational-clifford-4}
B_{\q,4}= 
\left[
\begin{array}{cccccccccccccccc}
\frac{1}{4} & 0 & 0 & 0 & 0 & 0 & 0 & 0 & 0 & 0 & 0 & 0 & 0 & 0 & 0 & 0 \\
\frac{1}{4} & \frac{1}{2} & 0 & 0 & 0 & 0 & 0 & 0 & 0 & 0 & 0 & 0 & 0 & 0 & 0 & 0 \\
\frac{1}{4} & 0 & \frac{1}{2} & 0 & 0 & 0 & 0 & 0 & 0 & 0 & 0 & 0 & 0 & 0 & 0 & 0 \\
\frac{1}{4} & 0 & 0 & \frac{1}{2} & 0 & 0 & 0 & 0 & 0 & 0 & 0 & 0 & 0 & 0 & 0 & 0 \\
\frac{1}{4} & 0 & 0 & 0 & \frac{1}{2} & 0 & 0 & 0 & 0 & 0 & 0 & 0 & 0 & 0 & 0 & 0 \\
\frac{1}{4} & 0 & 0 & 0 & 0 & \frac{1}{2} & 0 & 0 & 0 & 0 & 0 & 0 & 0 & 0 & 0 & 0 \\
\frac{1}{4} & 0 & 0 & 0 & 0 & 0 & \frac{1}{2} & 0 & 0 & 0 & 0 & 0 & 0 & 0 & 0 & 0 \\
\frac{1}{4} & \frac{1}{2} & \frac{1}{2} & \frac{1}{2} & \frac{1}{2} & \frac{1}{2} & \frac{1}{2} & 1 & 0 & 0 & 0 & 0 & 0 & 0 & 0 & 0 \\
\frac{1}{4} & 0 & 0 & 0 & 0 & 0 & 0 & 0 & \frac{1}{2} & 0 & 0 & 0 & 0 & 0 & 0 & 0 \\
\frac{1}{4} & 0 & 0 & 0 & 0 & 0 & 0 & 0 & 0 & \frac{1}{2} & 0 & 0 & 0 & 0 & 0 & 0 \\
\frac{1}{4} & 0 & 0 & 0 & 0 & 0 & 0 & 0 & 0 & 0 & \frac{1}{2} & 0 & 0 & 0 & 0 & 0 \\
\frac{1}{4} & \frac{1}{2} & \frac{1}{2} & \frac{1}{2} & 0 & 0 & 0 & 0 & \frac{1}{2} & \frac{1}{2} & \frac{1}{2} & 1 & 0 & 0 & 0 & 0 \\
\frac{1}{4} & 0 & 0 & 0 & 0 & 0 & 0 & 0 & 0 & 0 & 0 & 0 & \frac{1}{2} & 0 & 0 & 0 \\
\frac{1}{4} & \frac{1}{2} & 0 & 0 & \frac{1}{2} & \frac{1}{2} & 0 & 0 & \frac{1}{2} & \frac{1}{2} & 0 & 0 & \frac{1}{2} & 1 & 0 & 0 \\
\frac{1}{4} & 0 & \frac{1}{2} & 0 & \frac{1}{2} & 0 & \frac{1}{2} & 0 & \frac{1}{2} & 0 & \frac{1}{2} & 0 & \frac{1}{2} & 0 & 1 & 0 \\
\frac{1}{4} & 0 & 0 & \frac{1}{2} & 0 & \frac{1}{2} & \frac{1}{2} & 0 & 0 & \frac{1}{2} & \frac{1}{2} & 0 & \frac{1}{2} & 0 & 0 & 1 \\
\end{array}
\right]
\end{equation}

The lattices with basis $B_{\q,n}$ obtained from the orbit of the rational sub-group of the Clifford group via \cref{eq:orbit-lattice} are isometric to 
\begin{equation}
\label{eq:rational-basis}
\mathrm{BW}_{\q,n} = \frac{1}{2}
\left(
\begin{array}{c|c}
     \mathrm{Re}(B_{\mathbb{C}}^{\otimes (n-1)}) & -\mathrm{Im}(B_{\mathbb{C}}^{\otimes (n-1)}) \\
     \hline 
     \mathrm{Im}(B_{\mathbb{C}}^{\otimes (n-1)}) &  \mathrm{Re}(B_{\mathbb{C}}^{\otimes (n-1)}).
\end{array}
\right),
\end{equation}
for $n=2,3,4,5$ with appropriate rescaling of the inner product.
We checked this numerically using \cite{Nemo.jl-2017}.
We provide a Jupyter notebook to reproduce this calculation and to provide additional details on the inner product rescaling.
The lattices with bases $\mathrm{BW}_{\q,n}$ are known as the rational Barnes-Wall lattices.

The lattice in \cref{eq:orbit-lattice} can be constructed as an integer span of the orbit of an arbitrary vector, not just $|0\rangle$. For example, by considering the integer span of the orbit of $|000\rangle + |100\rangle$ under the rational sub-group of the Clifford group we get the lattice with basis 
\begin{equation}
\label{eq:basis-rational-clifford-3}
\tilde B_{\q,3}= 
\left[
\begin{array}{cccccccc}
\frac{1}{2} & 0 & 0 & 0 & 0 & 0 & 0 & 0 \\
\frac{1}{2} & 1 & 0 & 0 & 0 & 0 & 0 & 0 \\
\frac{1}{2} & 0 & 1 & 0 & 0 & 0 & 0 & 0 \\
\frac{1}{2} & 0 & 0 & 1 & 0 & 0 & 0 & 0 \\
\frac{1}{2} & 0 & 0 & 0 & 1 & 0 & 0 & 0 \\
\frac{1}{2} & 0 & 0 & 0 & 0 & 1 & 0 & 0 \\
\frac{1}{2} & 0 & 0 & 0 & 0 & 0 & 1 & 0 \\
\frac{1}{2} & 1 & 1 & 1 & 1 & 1 & 1 & 2 \\
\end{array}
\right].
\end{equation}
The rational subgroup of the Clifford group is a subgroup of the automorphism of the lattice with basis $\tilde B_{\q,3}$ by construction.

By direct calculation, we checked that the rational subgroup of the Clifford group is equal to the intersection of the automorphism groups of the lattices with bases $\tilde{B}_{\mathbb{Q},3}$ and $B_{\mathbb{Q},3}$. We computed the size of the rational subgroup of the Clifford group on three qubits to be $2,580,480$. We then found that the size of the intersection of the isomorphism groups of lattices with bases $\tilde{B}_{\mathbb{Q},3}$ and $B_{\mathbb{Q},3}$ is also $2,580,480$. Because both automorphism groups include the rational subgroup of the Clifford group, their intersection also includes the rational subgroup of the Clifford group as a subgroup. However, this subgroup has the same size as the entire group, which confirms the required result. We provide a Jupyter notebook to reproduce the calculation in Julia using~\cite{Nemo.jl-2017}.
 
\subsubsection{Qudit Clifford groups}
\label{sec:qudit-clifford-basis}
The qudit Clifford group on a $d$-dimensional qudit is generated by generalized $X,Z,H$ and $S$ gates \cite{Hostens2005}.
The generalized $H$ gate is also known as the $d$-dimensional discrete Fourier transform.
Multi-qudit Clifford groups are generated by the one-qudit generalized $X,Z,H,S$ and a sum gate $\ket{x}\ket{y} \mapsto \ket{x}\ket{(x+y) \text{ mod } d }$ for $x,y \in [d]$
which generalizes the $CX$ gate.
For our results to hold, we adjust the global phase of the generalized $H$ gate, so its entries are in $\q(d)$, similar to how we introduced $\tilde H$ so its entries were in $\q(i)$.
The center of the qudit Clifford groups for odd $d$ is $\{ \pm \zeta^j_d I : j \in [d] \}$.

The basis changes for one and two qutrits are:
\begin{equation}
\label{eq:basis-qutrit-clifford-1}
B^{(3)}_1 = 
\left[
\begin{array}{ccc}
\frac{-1}{3} \cdot \zeta_{3} + \frac{1}{3} & 0 & 0 \\
\frac{-1}{3} \cdot \zeta_{3} + \frac{1}{3} & 1 & 0 \\
\frac{-1}{3} \cdot \zeta_{3} + \frac{1}{3} & 0 & 1 \\
\end{array}
\right]
\end{equation}
 
\begin{equation}
\label{eq:basis-qutrit-clifford-2}
B^{(3)}_2 = 
\left[
\begin{array}{c|c|c|c|c|c|c|c|c}
\frac{1}{3} & 0 & 0 & 0 & 0 & 0 & 0 & 0 & 0 \\
\frac{1}{3} & \frac{2}{3} \cdot \zeta_{3} + \frac{1}{3} & 0 & 0 & 0 & 0 & 0 & 0 & 0 \\
\frac{1}{3} & 0 & \frac{1}{3} \cdot \zeta_{3} + \frac{2}{3} & 0 & 0 & 0 & 0 & 0 & 0 \\
\frac{1}{3} & 0 & 0 & \frac{1}{3} \cdot \zeta_{3} + \frac{2}{3} & 0 & 0 & 0 & 0 & 0 \\
\frac{1}{3} & 0 & 0 & 0 & \frac{1}{3} \cdot \zeta_{3} + \frac{2}{3} & 0 & 0 & 0 & 0 \\
\frac{1}{3} & \frac{2}{3} \cdot \zeta_{3} + \frac{1}{3} & \frac{1}{3} \cdot \zeta_{3} + \frac{2}{3} & \frac{2}{3} \cdot \zeta_{3} + \frac{4}{3} & \frac{2}{3} \cdot \zeta_{3} + \frac{4}{3} & 1 & 0 & 0 & 0 \\
\frac{1}{3} & 0 & 0 & 0 & 0 & 0 & \frac{1}{3} \cdot \zeta_{3} + \frac{2}{3} & 0 & 0 \\
\frac{1}{3} & \frac{4}{3} \cdot \zeta_{3} + \frac{2}{3} & 0 & \frac{1}{3} \cdot \zeta_{3} + \frac{2}{3} & \frac{2}{3} \cdot \zeta_{3} + \frac{4}{3} & 0 & \frac{1}{3} \cdot \zeta_{3} + \frac{2}{3} & 1 & 0 \\
\frac{1}{3} & \frac{4}{3} \cdot \zeta_{3} + \frac{2}{3} & \frac{1}{3} \cdot \zeta_{3} + \frac{2}{3} & \frac{2}{3} \cdot \zeta_{3} + \frac{4}{3} & \frac{1}{3} \cdot \zeta_{3} + \frac{2}{3} & 0 & \frac{1}{3} \cdot \zeta_{3} + \frac{2}{3} & 0 & 1 \\
\end{array}
\right]
\end{equation}

Note that $(B^{(3)}_1)^{\otimes 2}$ and $(B^{(3)}_2)$ generate different lattices, so qutrit case does not have the same nice tensor-product structure of basis matrices.

The basis change for one qudit with $d=5$ is 
\begin{equation}
\label{eq:basis-qupant-clifford}
B^{(5)}_1 = 
\left[
\begin{array}{c|c|c|c|c}
\frac{-2}{5} \cdot \zeta_{5}^{3} - \frac{2}{5} \cdot \zeta_{5}^{2} - \frac{1}{5} & 0 & 0 & 0 & 0 \\
\frac{-2}{5} \cdot \zeta_{5}^{3} - \frac{2}{5} \cdot \zeta_{5}^{2} - \frac{1}{5} & \frac{-3}{5} \cdot \zeta_{5}^{3} - \frac{1}{5} \cdot \zeta_{5}^{2} + \frac{1}{5} \cdot \zeta_{5} - \frac{2}{5} & 0 & 0 & 0 \\
\frac{-2}{5} \cdot \zeta_{5}^{3} - \frac{2}{5} \cdot \zeta_{5}^{2} - \frac{1}{5} & 0 & \frac{-3}{5} \cdot \zeta_{5}^{3} - \frac{1}{5} \cdot \zeta_{5}^{2} + \frac{1}{5} \cdot \zeta_{5} - \frac{2}{5} & 0 & 0 \\
\frac{-2}{5} \cdot \zeta_{5}^{3} - \frac{2}{5} \cdot \zeta_{5}^{2} - \frac{1}{5} & \frac{-6}{5} \cdot \zeta_{5}^{3} - \frac{2}{5} \cdot \zeta_{5}^{2} + \frac{2}{5} \cdot \zeta_{5} - \frac{4}{5} & \frac{-9}{5} \cdot \zeta_{5}^{3} - \frac{3}{5} \cdot \zeta_{5}^{2} + \frac{3}{5} \cdot \zeta_{5} - \frac{6}{5} & 1 & 0 \\
\frac{-2}{5} \cdot \zeta_{5}^{3} - \frac{2}{5} \cdot \zeta_{5}^{2} - \frac{1}{5} & \frac{-6}{5} \cdot \zeta_{5}^{3} - \frac{2}{5} \cdot \zeta_{5}^{2} + \frac{2}{5} \cdot \zeta_{5} - \frac{4}{5} & \frac{-3}{5} \cdot \zeta_{5}^{3} - \frac{1}{5} \cdot \zeta_{5}^{2} + \frac{1}{5} \cdot \zeta_{5} - \frac{2}{5} & 0 & 1 \\
\end{array}
\right]
\end{equation}

\section{Sizes of common Clifford groups}
\label{sec:clifford-group-sizes}
Using a well-known formula for the size of symplectic groups over finite field~\cite{o1978symplectic}, we get the following size of qubit and qudit Clifford groups:

\def\arraystretch{1.0}
\begin{equation}
\label{eq:clifford-sizes}
\begin{array}{ccr}
\text{Qudit dimension} & \text{Number of qudits} & \text{Clifford group size (up to phases)} \\
\hline 
\hline
2 & 1 & 24 \\
2 & 2 & 11\,520 \\
2 & 3 & 92\,897\,280 \\
\hline
3 & 1 & 216 \\
3 & 2 & 4\,199\,040 \\
\hline
5 & 1 & 3\,000 \\
\end{array}
\end{equation}

\end{document}